\documentclass[11pt]{article}
\usepackage{fullpage}
\usepackage{comment}
\usepackage{amsmath}
\usepackage{amsthm}
\usepackage{amssymb}
\usepackage{url}
\usepackage[final]{showkeys}
\usepackage{cleveref}
\usepackage{tikz}
\usepackage{lineno}

\usepackage{algorithmic}
\usepackage{floatpag}
\usepackage[ruled,vlined,commentsnumbered,titlenotnumbered]{algorithm2e}

\usepackage[toc,page]{appendix}

\newtheorem{theorem}{Theorem}
\newtheorem{lemma}{Lemma}

\newtheorem{claim}{Claim}

{\makeatletter
 \gdef\xxxmark{%
   \expandafter\ifx\csname @mpargs\endcsname\relax 
     \expandafter\ifx\csname @captype\endcsname\relax 
       \marginpar{xxx}
     \else
       xxx 
     \fi
   \else
     xxx 
   \fi}
 \gdef\xxx{\@ifnextchar[\xxx@lab\xxx@nolab}
 \long\gdef\xxx@lab[#1]#2{{\bf [\xxxmark #2 ---{\sc #1}]}}
 \long\gdef\xxx@nolab#1{{\bf [\xxxmark #1]}}
}

\newcommand{\poly}{\mbox{poly}}

\def \eps {\varepsilon}

\makeatletter
\newcommand{\removelatexerror}{\let\@latex@error\@gobble}
\makeatother

\begin{document}

\def \isnotin {\nsubseteq}

\def \eps {\varepsilon}

\title{\Large Subtree Isomorphism Revisited \thanks{A.A. and V.V.W. were supported by NSF Grants CCF-1417238 and CCF-1514339, and BSF Grant BSF:2012338. 
A.B. was supported by the NSF and the Simons Foundation; part of the work was done while the author was at the Thomas J. Watson Research Center.
T.D.H. was supported by the Carlsberg Foundation, grant no. CF14-0617. 
O.Z. was supported by BSF grant no.\ 2012338 and by
The Israeli Centers of Research Excellence (I-CORE) program (Center
No.\ 4/11).
}}
\author{Amir Abboud \\ Stanford University \\ \texttt{ \small abboud@cs.stanford.edu}
\and Arturs Backurs \\ MIT \\ \texttt{ \small backurs@mit.edu}
\and Thomas Dueholm Hansen \\ Aarhus University \\ \texttt{ \small tdh@cs.au.dk}
\and Virginia {Vassilevska Williams} \\ Stanford University \\ \texttt{ \small virgi@cs.stanford.edu}
 \and Or Zamir \\ Tel Aviv University \\ \texttt{\small orzamir@mail.tau.ac.il}
 }
\date{}
\maketitle

\begin{abstract} 
The \emph{Subtree Isomorphism} problem asks whether a given tree is contained in another given tree.
The problem is of fundamental importance and has been studied since the 1960s.
For some variants, e.g., \emph{ordered trees}, near-linear time algorithms are known, but for the general case truly subquadratic algorithms remain elusive.

Our first result is a reduction from the Orthogonal Vectors problem to Subtree Isomorphism, showing that a truly subquadratic algorithm for the latter refutes the Strong Exponential Time Hypothesis (SETH).

In light of this conditional lower bound, we focus on natural special cases for which no truly subquadratic algorithms are known. We classify these cases against the quadratic barrier, showing in particular that:

\begin{itemize}
\item
Even for binary, rooted trees, a truly subquadratic algorithm refutes SETH.
\item Even for rooted trees of depth $O(\log\log{n})$, where $n$ is the total number of vertices, a truly subquadratic algorithm refutes SETH.
\item
For every constant $d$, there is a constant $\eps_d>0$ and a randomized, truly subquadratic algorithm for degree-$d$ rooted trees of depth at most $(1+ \eps_d) \log_{d}{n}$.
In particular, there is an $O(\min\{ 2.85^h ,n^2 \})$ algorithm for binary trees of depth $h$.
\end{itemize}

Our reductions utilize new ``tree gadgets" that are likely useful for future SETH-based lower bounds for problems on trees. Our upper bounds apply a folklore result from randomized decision tree complexity.

\end{abstract}

\section{Introduction}


Trees are among the most frequently used and commonly studied objects in computer science. 
One of the most basic and fundamental computational problems on trees is whether one tree is contained in another, that is, can an isomorphic copy of $H$ be obtained by deleting nodes and edges of $G$. 
This problem is known under three names: \emph{Subtree Isomorphism}, \emph{Tree Pattern Matching} and \emph{Subgraph Isomorphism on Trees}.
There are a few variants of the problem, mainly determined by (1) whether the trees are rooted or unrooted, (2) whether their degrees are bounded, and (3) whether the trees are ordered, i.e. whether the order of the children of each node must be preserved by the isomorphism. In this paper we focus on the case of rooted, unordered trees with degrees bounded by a constant $d$.

Because of its fundamental importance, the time complexity of Subtree Isomorphism
has been studied since the 1960s, e.g. by Matula~\cite{Matula1968} and Edmonds (see \cite{Matula1978}).
The problem is an interesting special case of the Subgraph Isomorphism problem, studied extensively in theoretical computer science. 
Subgraph Isomorphism is well known to be NP-hard since it generalizes hard problems such as Clique \cite{Karp72}. 
It is notoriously difficult: unlike most natural NP-complete problems, it {\em requires} $2^{\omega(n)}$ time (under the exponential time hypothesis (ETH))~\cite{CPS15}.
Special cases of subgraph isomorphism, especially ones that are in P, have received extensive attention. 
A recent 85-page paper by Marx and Pilipczuk \cite{MP14} covers the case in which $H$ is of fixed constant size.
Besides fixing the size of $H$, there are other non-trivial ways to make the problem polynomial time solvable; Subtree Isomorphism is the earliest and arguably the most natural one.
Polynomial time algorithms were also obtained for biconnected outerplanar graphs \cite{Lingas89}, two-connected series-parallel graphs \cite{LP}, and more \cite{MT92,DLP00}, while it is known that further generalizations quickly become NP-hard, e.g., when $G$ is a forest and $H$ is a binary tree \cite{GJ}.

The problem is also of practical relevance, since it can model important applications in a wide variety of areas.
Subtree Isomorphism is at the core of many more expressive problems, such as \emph{Largest Common Subtree} \cite{KMY95,AHM00,ATMA14}, which generally ask: how ``similar" are two trees?
Application areas include computational biology \cite{ZS89}, structured text databases \cite{KM95}, and compiler optimization \cite{Tai79}.
Several definitions of tree-similarity have been proposed, and the search for fast algorithms for computing them, both in theory and in practice, has been ongoing for a few decades - see \cite{Bille05,Gallagher06,gusfield, Valiente13book} for surveys and textbooks.
We focus on Subtree Isomorphism, and then briefly discuss how the techniques introduced in this paper can be adapted to prove new results for the Largest Common Subtree problem as well.

\paragraph{Previous results.}
According to Matula \cite{Matula1978}, the first algorithms for Subtree Isomorphism were proposed in 1968 independently by Edmonds and Matula himself \cite{Matula1968}.
10 years later, Reyner \cite{Reyner77} and Matula \cite{Matula1978} showed that these algorithms run in polynomial time and the runtime is $O(n^{2.5})$. 
The algorithm executes many calls to a subroutine that solves maximum matching in bipartite graphs.
These result were for rooted trees, and later Chung \cite{Chung1978} showed that the same bounds can be achieved for unrooted trees.
In 1983, Lingas \cite{Lingas83} shaved a log factor, and
the most recent development was in 1999 by Shamir and Tsur \cite{ST99} who used the more recent randomized algorithms for bipartite matching \cite{Che97} to reduce the runtime to $O(n^{\omega})$ where $\omega < 2.373$ is the matrix multiplication exponent \cite{v12,legallmult}.

Interestingly, in the most basic case of rooted and constant degree trees, even the early algorithms run in $O(n^2)$ time, and the fastest known runtime is $O(n^2/\log{n})$ \cite{Lingas83, ST99}.
For comparison, when the trees are \emph{ordered},
a long line of STOC/FOCS papers \cite{Kosaraju89,DGM94,CH97,Indyk97,Indyk98,CH02} brought down the complexity of the problem from quadratic \cite{HO82} to $O(n\log{n})$ time \cite{CH03}.
It is natural to wonder whether the same improvements can be achieved in the case of unordered trees.


\paragraph{Main results.} 
Our main result is a conditional lower bound for Subtree Isomorphism.
We show that a truly subquadratic algorithm is unlikely, even on very restricted cases such as those of binary, rooted trees or rooted trees of depth $O(\log\log{n})$.
A matching upper bound, up to $n^{o(1)}$ factors, has been known since the 1960s (we briefly discuss this algorithm in Section~\ref{sec:UB}).


%
%
%
%
Our lower bounds are conditioned on the well-known Strong Exponential Time Hypothesis (SETH) of Impagliazzo, Paturi and Zane~\cite{IP01,IPZ01} which roughly states that as $k$ grows, $k$-SAT on $n$ variables requires $2^{(1-\eps)n}\textrm{ poly}(n)$ time for all $\eps>0$.
Our result for Subtree Isomorphism is the first ``SETH-hard" problem on \emph{trees}, which is an exciting addition to the diverse list\footnote{These are problems with $O(n^c)$ upper bounds for some $c>1$ and an $O(n^{c-\eps})$ algorithm, for some $\eps>0$, is known to refute SETH.} 
which already includes problems on vectors \cite{W04}, (general) graphs \cite{PW10, RV13, AV14,AGV15}, sequences \cite{AVW14, BI15, ABV15, BK15}, and curves \cite{Bring}.
Our ideas and constructions of ``tree gadgets" are useful for proving conditional lower bounds for other problems on trees. We demonstrate this with a lower bound for the Largest Common Subtree problem, discussed below.

\begin{theorem} 
\label{thm:lb}
For all $d\geq 2$, Subtree Isomorphism on two rooted, unordered trees of size $O(n)$, degree $d$, and height $h \leq 2 \log_d{n}+O(\log\log{n})$ cannot be solved in truly subquadratic $O(n^{2-\eps})$ time under SETH.
\end{theorem}

More generally, if the size of the smaller tree is $n$ and the bigger tree is $m$, then our lower bound says that $O(nm^{1-\eps})$ time refutes SETH.
We remark that since SETH is believed to hold even for randomized algorithms, our lower bound is also a barrier for truly subquadratic randomized algorithms.

To complement our lower bound, we proceed to tackle natural restrictions of the problem algorithmically. The most natural way to restrict tree inputs is to bound the degree or height.
%
Our lower bound leaves little room for improvement:
Even on binary trees of height $(2+o(1))\log{n}$ any algorithm must take quadratic time under SETH (note that the minimum height of a binary tree is $\log{n}$).

An intriguing case is when the trees are binary and almost complete, i.e., $d=2$ and $h = (1+o(1)) \log{n}$.
We are unable to show a super-linear lower bound in this case, nor are we able to obtain a deterministic algorithm that runs in truly subquadratic time. 
Nevertheless, we present a {\em randomized}, Las Vegas, algorithm that solves this case in truly subquadratic $O(n^{1.507})$ time.  Our algorithm solves more general cases:

\begin{theorem}\label{heightthm}
There is a randomized algorithm for rooted Subtree Isomorphism with expected running time $O(\min\{2.8431^h,n^2\})$ for trees $H$ and $G$ of size $O(n)$ and height at most $h$. In particular, the algorithm runs in time $O(n^{1.507})$ for trees of depth $(1+o(1))\cdot \log_2{n}$ and is truly subquadratic for trees of depth $h\leq 1.3267\cdot\log_2{n}$.
\end{theorem}

Our algorithm is simple, natural, and easy to implement.
Perhaps more interesting than the upper bound itself is that the technique we use to obtain it uses a technique from randomized decision tree complexity.

We also consider the case of ternary trees, providing a fast Las Vegas algorithm for it. Our approach is similar to that of the binary tree case. However, here we use a computer program to analyze the expected running time of the algorithm. 

\begin{theorem}
\label{thm:deg3}
	There is a randomized algorithm that can solve Subtree Isomorphism on two rooted \emph{ternary} trees of size $O(n)$ and height at most $h$  in expected $O\left(\min\left\{ 6.107^h , n^2 \right\}\right)$ time.
\end{theorem}

Finally, we generalize our algorithms to obtain truly subquadratic algorithms for rooted Subtree Isomorphism on trees with small height and constant degree $d$, for any $d \geq 2$. 

\begin{theorem} \label{constanddegree}
	There is a randomized algorithm that solves Subtree Isomorphism on two rooted trees of size $O(n)$, constant degree $d$, and height at most $h$  in expected time $$O\left(\min\left\{\left(d^2-\frac{1}{3} d+{2 \over 3}\right)^h , n^2 \right\}\right).$$
	In particular, the algorithm is strongly subquadratic for trees of height 
	$$
		h~\leq~ \left(\frac{\log(d^2)}{\log(d^2-\frac{1}{3} d+{2 \over 3})}-\epsilon\right) \cdot \log_d{n}~,
	$$
	for any constant $\epsilon>0$.
\end{theorem}

The bound in the above theorem is not tight for small $d$, as our algorithms for $d=2$ and $d=3$ show. 
For example, it is not subquadratic (on small depth trees) unless $d>3$. 
To obtain the upper bound, we prove a new randomized query complexity upper bound for bipartite perfect matching, which could be of independent interest (Lemma~\ref{query}).

This work is another example of a fine-grained study of the complexity of fundamental problems in P under natural parameterizations. This approach was formalized in two recent works~\cite{AVW15,GMN15}.

%
%
\paragraph{Techniques and other results.}

To prove our SETH hardness results we show reductions from Orthogonal Vectors to Subtree Isomorphism in Section~\ref{sec:LB}.
The reductions follow all previous SETH-hardness results in spirit, but require careful constructions of ``tree gadgets" that represent vectors, as well as techniques for combining the gadgets into two big trees $H$ and $G$ for which the existence of an orthogonal pair of vectors determines whether $H$ is contained in $G$.
Our reduction is clean and simple, but it gets more tricky when restricted to trees of constant degree.

Our reduction is easily modified to obtain similar lower bounds for related problems such as \emph{Largest Common Subtree} on two trees (LCST).
This problem is NP-hard when the number of trees is a parameter or when the two trees are labelled (and unrooted) \cite{ZSS92,ZJ94}, while some approximation and parameterized algorithms are known \cite{KMY95,ATMA14,AHM00}.
When the two trees are binary and unlabeled, the problem can be solved in quadratic time, and an adaptation of Theorem~\ref{thm:lb} shows that even when the height is $(1+o(1))\log{n}$, a truly subquadratic algorithm refutes SETH. 

\begin{theorem}
\label{thm:LCST}
For all $d\geq 2$, The Largest Common Subtree problem on two rooted trees of size $O(n)$, degree $d$ and height $h \leq \log_d{n}+O(\log\log{n})$ cannot be solved in truly subquadratic $O(n^{2-\eps})$ time under SETH.
\end{theorem}

Theorem~\ref{thm:LCST} is surprising when contrasted with our other results.
On the one hand, for arbitrary rooted trees with constant degrees, both Subtree Isomorphism and the harder-looking LCST have tight quadratic upper and (conditional) lower bounds.
On the other hand, we show that under the further restriction that the trees have small depth (as in Theorem~\ref{heightthm}), Subtree Isomorphism can be solved in truly subquadratic time, while by Theorem~\ref{thm:LCST} the LCST problem \emph{cannot}, under SETH.

%

We attribute our new algorithmic results to two ingredients.
The first important ingredient comes from our \emph{lower bounds}.
In particular we noticed that when the trees are binary and the depth is $(1+\eps)\log{n}$, it is difficult to implement our reductions. This turned our attention to finding upper bounds.
Knowing the hard cases thus allowed us to focus on the solvable cases.
This is an important byproduct of the recent research on conditional lower bounds in P.

The second ingredient was making a connection between this problem and a seminal result from randomized decision tree complexity \cite{SW86}. 
Our algorithm for binary (and ternary) trees is inspired by the following well-known result from complexity theory:
Given a formula represented by a complete AND-OR tree on $n$ leaves that represent the variables, can you evaluate the formula without looking at all the inputs?
The surprising fact is that this is possible with randomization: to evaluate a gate, we guess which child to check first at random, and if we see a $1$ input to an OR gate, or a $0$ input to an AND gate, we do not have to check the other child.
Therefore it is possible to evaluate the formula by only looking at $n^{1-\eps}$ inputs.
This result has found many applications in various areas of complexity theory, learning theory, and quantum query complexity \cite{ACRSZ10}.
%

\paragraph{Other related work.}
In the late 1980s, Subtree Isomorphism was considered from the viewpoint of efficient parallel algorithms. 
Lingas and Karpinski \cite{LK89} placed the problem in randomized $NC^1$.
Gibbons, Miller, Karp, and Soroker \cite{GMKS90} independently obtained the same result and also showed an $NC^1$ reduction from bipartite matching to Subtree Isomorphism. Their reduction takes a matching instance on $n$ nodes and produces trees on $\Omega(n^3)$ nodes, and therefore does not imply a lower bound on the time complexity of Subtree Isomorphism even assuming that current matching algorithms are optimal.
Note that any many-to-one reduction from matching (where the input is of size $\Omega(n^2)$) will generate trees of size $\Omega(n^2)$.
To get our quadratic lower bound we reduce from a different problem, namely Orthogonal Vectors.

Many related cases of the problem can be solved in near-linear time.
For example, when both trees have exactly the same size, we get the \emph{Tree Isomorphism} problem which was solved in $O(n)$ time by Hopcroft and Tarjan \cite{HT72}, and later other linear time algorithms were suggested (see \cite{DIR99} and the references therein).
Another example is the case of \emph{ordered trees}, meaning that there is an order among the children of a node that cannot be modified in the isomorphism.
Also, when a ``subtree" is defined to be a node and all its descendants, ``subtree" isomorphism can be solved in linear time \cite{Verma92}.


\section{SETH Lower Bounds}
\label{sec:LB}
The Strong Exponential Time Hypothesis (SETH) states that for every $\eps>0$ there exists a $k$ such that $k$-SAT on $n$ variables cannot be solved in $O(2^{(1-\eps)n}\poly n)$ time.
Williams \cite{W04} related SETH to a polynomial time problem called Orthogonal Vectors (OV). The inputs to OV are two lists of $N$ vectors in $\{0,1\}^D$ and the output is ``yes" if and only if there is a pair of vectors $\alpha,\beta$, one from each list, that are orthogonal, i.e. for all $i \in [D]$ either $\alpha[i]$ or $\beta[i]$ is equal to $0$.
Williams reduced CNF-SAT to OV so that if OV
 can be solved in $O(N^{2-\eps})$ time when $D=\omega(\log{N})$, for some $\eps>0$, then CNF-SAT on $n$ variables and $\poly~n$ clauses can be solved in $O(2^{(1-\eps')n}\poly~n)$ time for some $\eps'>0$, and SETH is false.

In this section we reduce CNF-SAT, via the Orthogonal Vectors (OV) problem, to different variants of the Subtree Isomorphism problem to prove our SETH-based lower bounds. 

\subsection{Hardness for Subtree Isomorphism}

\paragraph{A simpler reduction.}
We start with a ``warm-up" reduction that presents the high-level idea of our proofs.
In Theorem \ref{th:unbounded} below we reduce OV to Subtree isomorphism on trees with $n=O(ND)$ vertices, unbounded degree, and height $h=O(D)$. We later show how to change the construction to get trees with small constant degree and small height. 

\begin{theorem} \label{th:unbounded}
Orthogonal Vectors on two lists of $N$ vectors in $\{0,1\}^D$ can be reduced to Subtree Isomorphism on two trees of size $O(ND)$ and depth $O(d)$.
\end{theorem}

\begin{proof}
Let us denote the vectors of the first list by $A=\{ \alpha_1,\ldots,\alpha_N \}$ and of the second list by $B = \{ \beta_1,\ldots,\beta_N\}$ and recall that our goal is to find a pair of vectors $\alpha\in A, \beta\in B$ such that for every coordinate $i \in [D]$ either $\alpha[i]=0$ or $\beta[i]=0$.

The first ingredient in the reduction is to construct \emph{vector gadgets}.

For every vector in the first list $\alpha \in A$ we create a vector gadget: a tree $H_\alpha$ of size $O(D)$ as follows.
First, add a path $u_0 \to u_1 \to u_2 \to \cdots \to u_{D+2}$ and let $u_0$ be the root of $H_\alpha$.
Then, for each coordinate $i \in [D]$ we consider $\alpha[i]$ and if it is a $1$ we add a node $u_{i,1}$ to the tree $H_\alpha$ as the child of the node $u_i$, i.e. we add the edge $u_i \to u_{i,1}$. 
Otherwise, if $\alpha[i]=0$, the only child of $u_i$ will be $u_{i+1}$.

We now define the vector gadgets for the vectors in the second list. 
For every $\beta \in B$ we create a vector gadget: a tree $G_\beta$ of size $O(D)$ as follows.
The first step is similar, we add a path $v_0 \to v_1 \to v_2 \to \cdots \to v_{D+2}$ and let $v_0$ be the root.
The difference is in the second step.
For each coordinate $i \in [D]$, we consider $\beta[i]$ and if it is a $0$ we add a node $v_{i,0}$ to $G_\beta$ as the child of the node $v_i$, i.e. we add the edge $v_i \to v_{i,0}$. 

The following simple claim is the key to our reduction and explains our gadget constructions.  

\begin{claim}
\label{cl1}
$H_\alpha$ is isomorphic to $G_\beta$ iff $\alpha,\beta$ are orthogonal.
\end{claim}

\begin{proof}
For the first direction, assume that $\alpha,\beta$ are orthogonal and therefore for every $i \in [D]$ we know that either $\alpha[i]=0$ or $\beta[i]=0$.
We will define a mapping $f$ from $H_\alpha$ to a subgraph of $G_\beta$ such that if $\{u,v\}$ is an edge in $H_\alpha$ then $\{ f(u),f(v)\}$ is an edge in $G_\beta$.
First, we map the roots and paths to each other, by setting $f(u_i)=v_i$ for all $i \in \{0,\ldots,D+2\}$.
Then, we consider every $i \in [D]$ for which $\alpha[i]=1$ and map $u_{i,1}$ to the node $v_{i,0}$ in $G_\beta$.
We are guaranteed that $v_{i,0}$ exists because if $\alpha[i]=1$ then $\beta[i]$ must be $0$, by the orthogonality of the vectors.
It is easy to check that two neighbours in $H_\alpha$ are mapped to two neighbours in $G_\beta$.

For the other direction, assume $H_\alpha$ is isomorphic to a subgraph of $G_\beta$, and let $f$ be the mapping.
First, note that $u_0$ must be mapped to $v_0$ since these are the roots of the two trees.
Then we observe that $u_{D+2}$ must be mapped to $v_{D+2}$ and the path $u_0\to \cdots \to u_{D+2}$ must be mapped to the path $v_0 \to \cdots \to v_{D+2}$ since these are the only paths of length at least $(D+2)$ in the trees.
Now, let $i \in [D]$ be such that $\alpha[i]=1$ and note that $u_i$ must have degree $3$ in this case, which implies that $f(u_i)=v_i$ must also have degree at least $3$ in $G_\beta$, which implies that the node $v_{i,0}$ must exist, and $\beta[i]=0$.
Thus, whenever $\alpha[i]=1$ it must be the case that $\beta[i]=0$, and the vectors are orthogonal.
\end{proof}

The final step is to combine the vector gadgets into two trees $H,G$ in a way such that $H$ is isomorphic to a subtree of $G$ if and only if there is a pair of orthogonal vectors within our two lists.

To this end, we define a special vector $\gamma = \vec{0}$ to be the all-zero vector in $D$ dimensions.
By Claim~\ref{cl1}, for any vector $\beta \in \{0,1\}^D$, we have that $H_\beta$ is isomorphic to a subtree of $G_\gamma$.

We are now ready to define the trees $H$ and $G$ of size $O(ND)$.

$G$ will be composed of a root node $g$ of degree $(2N-1)$ that has $G_{\beta_j}$ as a child for every $\beta_j \in B$, in addition to $(N-1)$ distinct $G_\gamma$ gadgets.
 That is, first, for each $j \in [N]$ add the vector gadget $G_{\beta_j}$ to $G$ and add the edge $g \to v_0$ where $v_0$ is the root of $G_{\beta_j}$.
And then, we add $(N-1)$ trees $G_\gamma^{(1)},\ldots, G_\gamma^{(n-1)}$ to $G$ and for each $j \in [N-1]$ we add the edge $g\to v_0$ where $v_0$ is the root of $G_\gamma^{(j)}$.

$H$ will be constructed in a similar way, except we do not add the $\gamma$ vector gadgets. 
Create a root node $h$ of degree $N$ that has $H_{\alpha_j}$ as a child for every $\alpha_j \in A$.
As in the definition of $G$, we add edges $h \to u_0$ where $u_0$ is the root of $H_{\alpha_j}$, for every $j \in [N]$.

Before proving the correctness of the reduction, note that the size of each tree is indeed $O(ND)$ since each gadget has size $O(D)$ and we are combining $O(N)$ gadgets into our trees $H,G$.
To conclude the proof, we claim that $H$ is isomorphic to a subgraph of $G$ iff there is a pair of orthogonal vectors.

\begin{claim}
In the above reduction, $H$ is isomorphic to a subtree of $G$ iff there is a pair $\alpha \in A, \beta \in B$ of orthogonal vectors.
\end{claim}

\begin{proof}
For the first direction, assume that there is a pair of orthogonal vectors $\alpha \in A, \beta \in B$ and we will show that $H$ is isomorphic to a subtree of $G$.
Consider the mapping which maps $H_\alpha$ to $G_\beta$ as in Claim~\ref{cl1}, and then for each of the $(N-1)$ $H_{\alpha'}$ subtrees, for $\alpha' \neq \alpha$, we map it to a different $G_\gamma$ subtree of $G$.
Finally, the root $h$ is mapped to $g$.
It is easy to check that neighbours in $H$ are mapped to neighbours in $G$.

For the other direction, assume that $H$ is isomorphic to a subgraph of $G$ and let $f$ be the corresponding mapping.
We know that $f(h)=g$ and for each vector gadget $H_{\alpha_j}$ in $H$, its image using our mapping $f$ must be entirely contained in exactly one vector gadget $G_{x}$ in $G$, where $x \in B \cup \{ \gamma \}$.
Moreover, two gadgets $H_{\alpha},H_{\alpha'}$ cannot be mapped to the same gadget $G_x$.
There are $N$ $H_\alpha$ gadgets but only $(N-1)$ $G_\gamma$ gadgets, thus, by the pigeonhole principle, there must be at least one $\alpha \in A$ for which $H_\alpha$ is mapped to a gadget $G_{x}$ for $x \neq \gamma$, i.e., $x = \beta$ for some $\beta \in B$.
We conclude that there is a mapping from $H_\alpha$ to $G_\beta$ in which every two neighbours are mapped to neighbours, that is, that $H_\alpha$ is isomorphic to a subgraph of $G_\beta$, which, by Claim~\ref{cl1}, implies that $\alpha \in A,\beta \in B$ are orthogonal.
\end{proof}
\end{proof}

\paragraph{Shorter Vector Gadgets.}
Next, we show how our reductions can be implemented with trees of smaller depth, by introducing a new construction of vector gadgets.
We will use these gadgets in our final reductions that prove Theorems~\ref{thm:lb} and~\ref{thm:LCST}.
%
%
%

\begin{lemma}  \label{vector_gadgets}
	Given two vectors $\alpha,\beta \in \{0,1\}^D$ we can construct two binary rooted trees $H_\alpha,G_\beta$ of depth $3\log_2(D)+O(1)$ in linear time, such that $H_\alpha$ is isomorphic to a subtree of $G_\beta$ if and only if $\alpha,\beta$ are orthogonal.
\end{lemma}
\begin{proof}
Our constructions will involve careful combinations of ``index gadgets", which are defined as follows.
For a sequence of $\ell$ binary values $b_1, b_2, \ldots, b_l$, we define a tree ``index gadget" $Q_{b_1, b_2, \ldots, b_l}$ (think of $\ell$ as being $\lceil \log_2(D+1)\rceil$ and think of $b_1, b_2, \ldots, b_l$ as bits representing an index in $[D]$) to be composed of a path $z_1 \to z_2 \to ... \to z_l$ of length $l$, in which $z_1$ is the root, and for all $i\in[l]$ we attach a child $z_{i,1}$ to $z_i$ if and only if $b_i=1$.
That is, our index gadget $Q_{b_1, b_2, \ldots, b_l}$ is representing the index in the natural way: the edge $z_i \to z_{i,1}$ will exist if and only if $b_i=1$.

	Our first ``vector gadget" $H_\alpha$ is constructed as follows. 
	First, we build a complete binary tree with $D$ leaves $u_1, u_2, \ldots, u_D$ where the subtree at each leaf $u_i$ will encode the entry $\alpha[i]$ using our ``index gadgets". 
	We assume that every index $i\in [D]$ can be represented by $l=\lceil \log_2(D+1)\rceil$ bits and we let ${\bar i}$ denote this representation and let ${\bar i}^S$ denote the binary sequence obtained by flipping each bit of ${\bar i}$.
	For each node $u_i$ we will attach three gadgets, one after the other: first we will attach the $Q_{\bar i}$ index gadget, then we follow it by the $Q_{{\bar i}^S}$ index gadget, and finally we append a path of length either $2$ or $3$ -- depending on $\alpha[i]$.
	The necessity of this complicated encoding will become clear in the proof of correctness below.
	More formally, we first attach $u_i \to Q_{\bar i}$, then we let $z_l'$ denote the node of $Q_{\bar i}$ corresponding to $z_l$ in the above construction (i.e. the last node on the path), and attach $z_l' \to Q_{{\bar i}^S}$.
	Then, similarly, we let $z_l''$ be the node of $Q_{{\bar i}^S}$ which corresponds to $z_l$ in the above construction (i.e. the last node on the path), and we either attach three nodes $z_l''\to a_i \to b_i \to c_i$ if $\alpha[i]=1$, or we attach only two nodes $z_l'' \to a_i \to b_i$.
	
	The second ``vector gadget" $G_\beta$ is constructed in the same way except that we attach a path of length $3$ if $\beta[i]=0$ (as opposed to $1$) and attach a path of length $2$ if $\beta[i]=1$. 
	By construction, the depth of both trees is $3\log_2(D)+O(1)$ as claimed.

	To complete the proof we show that $H_\alpha$ is isomorphic to a subtree of $G_\beta$ iff $\alpha \cdot \beta=0$. 
	The first direction is easy: if the vectors are orthogonal then the natural mapping from $H_\alpha$ to $G_\beta$ that follows from our construction shows the isomorphism: map the binary trees on top to each other so that the $u_i$'s are mapped to each other, then map the attached $Q_{\bar i}\to Q_{{\bar i}^S}$ subtrees to each other, and finally, we can map the paths $a_i \to b_i \to c_i$ (if $\alpha[i]=1$) or    $a_i \to b_i $ (if it is $0$) to each other since in the first case $\beta[i]$ must be zero and $c_i$ will also exist in $G_\beta$.
	
 It remains to show that if $H_\alpha$ is isomorphic to a subtree of $G_\beta$, then $\alpha \cdot \beta=0$. Our index gadgets $Q_{\bar i}$ and $Q_{{\bar i}^S}$ will play a crucial role in this part, as they will show that in any mapping between the leaves of the complete tree we must map $u_i$ in $H_\alpha$ to $u_i$ in $G_\beta$ or else the index gadgets will not map into each other properly.
 We claim that for any two indices $i,j \in [D]$ we have that $i=j$ if and only if both $Q_{\bar i}$ is contained in $Q_{\bar j}$ \emph{and} $Q_{{\bar i}^S}$ is contained in $Q_{{\bar j}^S}$.
  This is true because of the following observation: $Q_{\bar x}$ is isomorphic to a subtree of $Q_{\bar y}$ iff the set of positions in $\bar x$ with $1$ is a subset of the set of positions of $\bar y$ with $1$.
  Therefore, any mapping from $H_\alpha$ to a subtree of $G_\beta$ must map the path representing $\alpha[i]$ to the path representing $\beta[i]$, for all $i \in [D]$. By construction, this can only happen if $\alpha\cdot \beta =0$.
  \end{proof}

\paragraph{Constant Degree Trees.}
Perhaps the most challenging element towards the proof of Theorem~\ref{thm:lb} is the combination of all the vector gadgets into two big trees, \emph{without using large degrees}.

To see the difficulty, recall the reduction in the proof of Theorem~\ref{th:unbounded}: in both trees, we added all $X$ vector gadgets as children of a root of degree $X$. 
By doing so we have essentially allowed the isomorphism to pick \emph{any} matching between the gadgets.
Combined with the auxiliary gadgets that we added, this allowed us to show that the final two trees are a ``yes" instance of Subtree Isomorphism if and only if the original vectors contained an orthogonal pair.
However, when the trees have constant degree (say, binary) it is much harder to combine the vector gadgets into two trees such that any matching between the gadgets can be chosen by the isomorphism.
A natural approach would be to add the gadgets at the leaves of a complete binary tree. 
One reason this does not work is that any isomorphism must map the first and second gadgets to adjacent gadgets in the second tree -- that is, only special kinds of matchings can be ``implemented".

We overcome this difficulty with a two-level construction that allows the isomorphism to pick exactly one gadget from each of the two trees and ``match" them, while all the other gadgets do not affect the outcome.  

\begin{theorem} \label{th:bounded}
	Given sets of vectors $A,B$, we can construct two rooted trees $H=H(A)$ and $G=G(B)$ such that the following properties hold. 
	\begin{enumerate}
		\item The number of nodes in both trees and the construction time is upper bounded by $O(ND)$. 
		\item The degree of both trees is upper bounded by $d$. 
		\item The depth of both trees is upper bounded by $2\log_d(N)+O(\log D)$.
		\item $H$ is isomorphic to a subtree of $G$ iff there are $\alpha \in A$ and $\beta \in B$ with $\alpha \cdot \beta=0$.
	\end{enumerate}
\end{theorem}
\begin{proof}
Let $\{H_\alpha\}_{\alpha \in A}=\{H_{\alpha_i}\}_{i \in [N]}$ and $\{G_\beta\}_{\beta \in B}=\{G_{\beta_i}\}_{i \in [N]}$ be the two sets of vector gadgets corresponding to the vectors of $A$ and $B$ that are obtained by the construction in Lemma~\ref{vector_gadgets}. 
We will now combine these vector gadgets into two big trees $H$ and $G$, which will be constructed quite differently from each other.

Assume that $\log_d(N)$ is an integer, otherwise add dummy vectors to increase $N$. 
The first tree $H$ will be composed of a complete $d$-ary tree with $N$ leaves $u_1, u_2, \ldots u_N$, followed by a path of length $\log_d(N)+1$, followed by the vector gadgets $H_{\alpha_i}$. 
More formally, for every $i \in [N]$ we add: $$ u_i \to h_{i,1} \to h_{i,2} \to \ldots \to h_{i,\log_d(N)+1} \to H_{\alpha_i}.$$

To construct the second tree $G$ we need to construct vector gadgets $G_\gamma$ corresponding to the all-zero vector $\gamma=\vec{0}$ of length $D$.
As before, we start with a complete $d$-ary tree with $N$ leaves $v_1, v_2, \ldots v_N$ and attach a path of length $\log_d(N)+1$ to each leaf, except for $v_N$ which will be treated differently. Then, we attach a copy of $G_\gamma$ at the end of each one of these paths, that is $N-1$ copies in total.
Formally, for every $i=1,\ldots,N-1$ we add:
	$$
		v_i \to h_{i,1} \to h_{i,2} \to \ldots \to h_{i,\log_d(N)+1} \to G_\gamma.
	$$
	Note that none of the vectors in the second list are encoded in this part of $G$ and they will appear now in the subtree rooted at $v_N$ which we construct next.
	Rooted at $v_N$, we add another complete $d$-ary tree with $N$ leaves $v_1', v_2', \ldots v_N'$, and then attach the vector gadgets right after these leaves. That is, for every $i \in [N]$ we add:  $v_i' \to G_{\beta_i}$.
	 
	 This finishes the construction of $H$ and $G$ and the first two properties are immediate. 
	 The third property follows from Lemma \ref{vector_gadgets}, and we now turn to proving the fourth property which is the correctness of our construction.

	 	\begin{claim}
		There is a pair of vectors $\alpha \in A$ and $\beta \in B$ with $\alpha\cdot \beta =0$ if and only if  $H$ is isomorphic to a subtree of $G$.
	\end{claim}
\begin{proof}
For the first direction, let $\alpha_i$ and $\beta_j$ be a pair of orthogonal vectors and we will show that $H$ is contained in $G$. 
First, consider the rearrangement of $H$ so that the rightmost leaf of the complete $d$-ary tree (where $u_N$ used to be) is $u_i$, the node to which the vector gadget $H_{\alpha_i}$ is attached.
We claim that all vector gadgets in $H$ can now be properly mapped to subtrees of $G$, without rearranging the $v_i$ nodes in $G$.
To see this, first note that all vector gadgets $H_{\alpha_{x}}$ for $x \neq i$ will be paired up with the $G_\gamma$ vector gadgets, and by Lemma \ref{vector_gadgets} and the fact that $\gamma$ is orthogonal to any vector, we know that there is a proper mapping.
Then, it remains to show that the subtree of $H$ rooted at $u_i$ is contained in the subtree of $G$ rooted at $v_N$, which follows because we can map the vector gadget $H_{\alpha_i}$ to the vector gadget $G_{\beta_j}$ since $\alpha_i\cdot\beta_j=0$.


For the second direction, assume that there is a mapping from $H$ to a subtree of $G$ and we will show that there must exist a pair of orthogonal vectors.
First, note that under this mapping, there is some $i \in [N]$ such that $u_i$ is mapped to $v_N$.
By construction of the subtree rooted at $v_N$, this means that the vector gadget $H_{\alpha_i}$ must be mapped into one of the vector gadgets $G_{\beta_j}$ for some $j \in [N]$, and not into $G_\gamma$.
By Lemma~\ref{vector_gadgets}, this can only happen if $\alpha_i \cdot \beta_j=0$.
	\end{proof}
\end{proof}

Theorem~\ref{th:bounded} and the connection between SETH and OV of Williams~\cite{W04} imply Theorem~\ref{thm:lb} from the introduction.

\subsection{Hardness for Largest Common Subtree}




Next, we prove a lower bound for the Largest Common Subtree (LCST) problem, which is a generalization of Subtree Isomorphism.
Although the reductions above already imply a quadratic lower bound for LCST, we will now optimize these reductions and prove a stronger hardness result: we will show that even on binary trees of depth $(1+o(1))\log{n}$ the LCST cannot be computed in truly subquadratic time.
This will show an interesting gap between LCST and Subtree Isomorphism, since the latter can be solved in truly subquadratic time on such trees - we present such upper bounds in Section~\ref{sec:UB}.
Our strengthened hardness result gives an explanation for why we are not able to extend our upper bounds to LCST: such extensions would refute SETH.
The next theorem implies Theorem~\ref{thm:LCST} from the introduction.

\begin{theorem}
If for some $\eps>0$, the Largest Common Subtree problem on two trees size $n$ can be solved in $O(n^{2-\eps})$ time, then Orthogonal Vectors on $N$ vectors in $\{0,1\}^D$ can be solved in $O(N^{2-\eps}\cdot D^{O(1)})$ time. The trees produced in the reduction from the Orthogonal Vectors problem have degree $d$ and height at most $\log_d(N)+O(\log D)$ for arbitrary $d \geq 2$.
\end{theorem}

\begin{proof}
We note that the construction provided in Theorem \ref{th:bounded} is not sufficient for our purposes because the height of the produced trees is $2\log_d(N)+O(\log D)$, which is larger than what we want. 
We will use the more expressive nature of LCST to implement our reduction with smaller height.
%


To achieve smaller height, we will try to implement vector gadgets such that the largest common subtree of two gadgets would be of a certain fixed size $E$ if the vectors are not orthogonal, while it will be of a larger size $E'>E$ if the vectors are orthogonal.
This trick was introduced by Backurs and Indyk in their reduction to Edit-Distance \cite{BI15} and later used in the reductions to LCS \cite{ABV15}.
Here, we carefully implement such gadgets with degree $d$ trees of small height instead of sequences. WLOG, we can assume that all vectors in $A$ start with $1$ and all vectors in $B$ start with $0$. If it is not so, we can add an extra coordinate at the beginning of every vector and set the entry accordingly. This does not change the answer to the problem (whether there are two orthogonal vectors). Also, we assume that all vectors in $A$ have the same number of entries equal to $1$. If it is not so, we can subdivide the set $A$ into smaller sets so that every set contain vectors with the same number of entries equal to $1$. Then we run the reduction on every subset of $A$ and $B$. This increase the runtime to solve the Orthogonal Vectors problem by a factor of $D+1$ but we are fine with that.

For each vector in the first list, $\alpha \in A$, we construct a vector gadget $H_\alpha$ as follows. Let $H_\alpha'$ be the vector gadget constructed in Lemma \ref{vector_gadgets} corresponding to vector $\alpha \in A$. Then $H_\alpha$ is equal to $r \to root(H_\alpha')$ for some vertex $r$, which is the root of $H_\alpha$.

For each vector in the second list, $\beta \in B$, we construct a vector gadget $G_\beta$ as follows. Let $\delta$ be a vector with $D$ coordinates. The first entry is equal to $1$ and the rest of entries are equal to $0$. Let $G_\beta'$ be the vector gadget constructed in Lemma \ref{vector_gadgets} corresponding to vector $\beta \in B$. Then we obtain $G_\beta$ by choosing a vertex $r$ to be its root and adding $r \to G_\delta'$ and $r \to G_\beta'$.

The main idea behind this construction is that, when matching $H_\alpha$ and $G_\beta$, one has a choice: either match $H'_\alpha$ to $G'_\delta$ (giving a fixed score, independent of $\alpha$), or match it to $G'_\beta$ (and the score then depends on the orthogonality of $\alpha$ and $\beta$.)
We make this argument formal in the next lemma.
Let $E'$ denote the size of $H_\alpha$ for $\alpha \in A$, which is independent of $\alpha$ since all vectors in $A$ contain the same number of $1$'s.
Let $E=E'-1$.

\begin{lemma}
\label{cl:LCSTmain}
The largest common subtree of $H_\alpha$ and $G_\beta$ is of size $E'=|H_\alpha|$ if $\alpha,\beta$ are orthogonal and it is of size $E=E'-1$ otherwise. We have that the size of $H_\alpha$ and $H_{\alpha'}$ are equal $|H_\alpha|=|H_{\alpha'}|$ for all $\alpha,\alpha' \in A$.
\end{lemma}
\begin{proof}
First, if $\alpha,\beta$ are orthogonal, then by Lemma~\ref{vector_gadgets} we have that $H_\alpha$ is isomorphic to a subgraph of $G_\beta$ and the LCST has size $E'$.


For the second case, assume that $\alpha,\beta$ are not orthogonal.
We first remark that there is a common subtree of size $E'-1$: Let $\alpha'$ denote $\alpha$ where we set the first coordinate of $\alpha$ (which is equal to $1$) to $0$, then $H'_{\alpha'}$ is a subtree of $H'_{\alpha}$ of size $|H'_{\alpha'}|=E'-1$, and by Lemma \ref{vector_gadgets}, it is also a subtree of $G'_{\delta}$ because $\alpha' \cdot \delta=0$. 
It remains to show that we cannot map the entire tree $H_{\alpha}$ to a subtree of $G_{\beta}$, which follows because $H'_{\alpha}$ is neither isomorphic to a subtree of $G'_{\delta}$ (since $\alpha \cdot \delta=1$) nor to a subtree of $G'_{\beta}$ (since $\alpha \cdot \beta\neq 0$).
\end{proof}


We are now ready to present the final trees $H,G$.
We construct $H$ as follows. First, we build a complete $d$-ary tree with $N$ leaves $h_1,\ldots,h_N$ at the lowest level. For every $j \in [N]$, we add $h_j \to H_{\alpha_j}$, where $A=\{\alpha_1, \ldots, \alpha_N\}$. Similarly we construct $G$. Take a complete $d$-ary tree with leaves $g_1, \ldots, g_N$ at the lower level. For every $j \in [N]$, we add $g_j \to G_{\beta_j}$, where $B=\{\beta_1, \ldots, \beta_N\}$.

\begin{theorem}
The Largest Common Subtree of $H$ and $G$ is of size at most $(2N-1)+(N \cdot E)$ if there is no pair of orthogonal vectors, and is at least $(2N-1)+(N \cdot E+1)$ otherwise.
\end{theorem}
\begin{proof}
We must map the nodes $h_i$ for every $i\in [N]$ to nodes $g_{\pi(i)}$, for some permutation $\pi:[N]\to[N]$.
Notice, however, that $\pi$ cannot be an arbitrary permutation since, e.g. $\pi(1)=\pi(2) \pm 1$ (the permutation must be implemented by swapping children in a complete binary tree.)

On the one hand, the total size of the common subtree can be upper bounded by the size of a complete binary tree with $N$ leaves, plus $\sum_{i=1}^N LCST(H_{\alpha_i}, G_{\beta_{\pi(i)}})$, for an arbitrary permutation $\pi$.
If there is no pair of orthogonal vectors, then by Lemma~\ref{cl:LCSTmain},  the latter sum is exactly $N \cdot E$, and the total size is bounded by $(2N-1)+N \cdot E$.

On the other hand, if there is an orthogonal pair $\alpha_i, \beta_j$, we can take any mapping in which $h_i$ is mapped to $g_j$ while the other $h_{x}$'s are mapped arbitrarily to different $g_{y}$'s.
This induces some permutation $\pi:[N]\to[N]$ so that $h_x$ is mapped to $g_{\pi(x)}$.
Since $\alpha_i\cdot\beta_j=0$, Lemma~\ref{cl:LCSTmain} implies that
this mapping can be completed to a mapping of score 
$$
(2N-1) + \sum_{v=1}^N LCST(H_{\alpha_v}, G_{\beta_{\pi(v)}}) ~\geq 
(2N-1)+(N-1)\cdot E + (E+1) ~=
(2N-1)+(N \cdot E+1)~.
$$
\end{proof}
\end{proof}

\section{Algorithms}
\label{sec:UB}
In this section we present new algorithms for Subtree Isomorphism on rooted trees with vertices of bounded degree. Edmonds and Matula independently described a procedure for reducing the rooted Subtree Isomorphism problem to a polynomially bounded collection of recursively smaller Subtree Isomorphism problems, and how to combine the answers by solving a maximum bipartite matching problem (see \cite{Matula1978}). We follow the same approach but focus on the case where the degrees are bounded by a constant.

Given two rooted trees $H$ and $G$, we want to decide whether $H$ is isomorphic to a subtree of $G$ where the root of $H$ maps to the root of $G$. Let $H_1,H_2,\ldots,H_k$ and $G_1,G_2,\ldots,G_\ell$ be the subtrees of $H$ and $G$, respectively, with roots that are children of the root of $H$ and the root of $G$. Let $\mathcal{G}$ be a bipartite graph with vertex set $\mathcal{V} = \{u_1,\ldots,u_k\} \cup \{v_1,\ldots,v_\ell\}$, and let $(u_i,v_j)$ be an edge of $\mathcal{G}$ if and only if $H_i$ is isomorphic to a subtree of $G_j$. Then $H$ is isomorphic to a subtree of $G$ if and only if $\mathcal{G}$ contains a matching of size $k$. The Edmonds-Matula procedure constructs the graph $\mathcal{G}$ by recursion and then solves the maximum bipartite matching problem on $\mathcal{G}$.

Designing similar algorithms for rooted Subtree Isomorphism thus involves two challenges: constructing $\mathcal{G}$ and solving the maximum bipartite matching problem on $\mathcal{G}$. The currently fastest randomized algorithm for the maximum bipartite matching problem is due to Mucha and Sankowski \cite{MS04} and runs in expected time $O((k+\ell)^\omega)$, where $\omega < 2.373$ is the matrix multiplication exponent. Improving this algorithm is itself a challenging open problem.

For constructing the graph $\mathcal{G}$, it is not hard to see that any deterministic algorithm needs to know all edges of $\mathcal{G}$. For randomized algorithms, however, it is not always necessary to know for every pair $u_i,v_j$ whether the edge $(u_i,v_j)$ is in the graph.
The expected number of node pair queries (``is the pair an edge in the graph?'') that a randomized algorithm needs to make in order to be able to determine whether a perfect matching exists, is known as the {\em randomized query complexity} (or decision tree complexity) of bipartite perfect matching.
It is an easy exercise to check that
the randomized query complexity of the problem is $\Omega(k \ell)$.
Estimating the exact number of queries is, however, not straightforward.
It is not even clear whether $k\ell$ queries are necessary in expectation, or whether $(1-\eps)k\ell$ queries might be sufficient for some $\eps>0$.
Factoring this into the analysis of the maximum bipartite matching algorithm complicates things further.

To simplify things, we
restrict our attention to the case where the degrees of the trees are bounded by a constant. In this case we can check in constant time whether $\mathcal{G}$ contains the desired perfect matching, once a sufficient number of edge queries have been made. We can thus focus solely on the randomized query complexity of the bipartite matching problem and its use in recursive algorithms for the Subtree Isomorphism problem.

It is easy to show that in this case the algorithm of Edmonds and Matula runs in time $O(mn)$, where $|H|=m$ and $|G|=n$. The same algorithm is also able to handle labelled vertices, i.e., each vertex has a label and the labels of $H$ are required to match the labels of the subtree of $G$. Moreover, the algorithm can solve the largest common subtree problem in $O(mn)$ time as well. This is done by recursively assigning a weight to every edge $(u_i,v_j)$ of $\mathcal{G}$ equal to the size of the largest common subtree of $H_i$ and $G_i$, and then asking for the matching of largest weight. (We refer to the appendix for a short complexity analysis and further description of these algorithms.) Our lower bounds from theorems \ref{thm:lb} and \ref{thm:LCST} are thus tight for trees of constant degree.

For the remainder of the section we restrict our attention to trees of constant degree $d$ and height $h$. We first introduce a randomized algorithm that solves the binary problem in expected time $O(\min\{2.8431^h,mn\})$. For comparison, the corresponding upper bound by Edmonds and Matula \cite{Matula1978} is $O(\min\{4^h,mn\})$, i.e., their algorithm makes four recursive calls at each level of the tree. In particular our algorithm is truly subquadratic when $h<1.3267\log_2{n}$. For $d=3$ we give a similar, but more complicated case analysis showing that the problem can be solved in expected time $O(\min\{6.107^h,mn\})$, improving the straightforward $O(\min\{9^h,mn\})$ bound by Edmonds and Matula. For $d > 3$ we introduce a randomized algorithm with expected running time upper bounded by $O(\min\{(d^2-\frac{1}{3} d+{2 \over 3})^h,mn\})$.

\subsection{A faster algorithm for binary trees}\label{sec:binary_upper}
For trees with degree at most two, the Edmonds-Matula procedure can be interpreted as follows. Let $H_L$ and $H_R$ be the left and right subtrees of $H$, and let $G_L$ and $G_R$ be the left and right subtrees of $G$. $H$ is isomorphic to a subtree of $G$ if and only if one of the following two conditions are true:
\begin{enumerate}
\item $H_L$ is isomorphic to a subtree of $G_L$, and $H_R$ is isomorphic to a subtree of $G_R$.
\item $H_L$ is isomorphic to a subtree of $G_R$, and $H_R$ is isomorphic to a subtree of $G_L$.
\end{enumerate}
Each case can be checked with two recursive calls, and checking whether $H$ is isomorphic to a subtree of $G$ can thus be done with at most four recursive calls, giving an $O(4^h)$ upper bound.

Observe that if $H_L$ is not isomorphic to a subtree of $G_L$, then there is no reason to check whether $H_R$ is isomorphic to a subtree of $G_R$. Similarly, if the algorithm concludes that the first condition is met, then there is no reason to check the second condition since we already know that $H$ is isomorphic to a subtree of $G$. Based on these observations, we introduce a simple randomized variant of the algorithm that achieves a significantly better running time by saving recursive calls: Swap $H_L$ and $H_R$ with probability $1/2$, and swap $G_L$ and $G_R$ with probability $1/2$. Then run the Edmonds-Matula algorithm, but do not perform unnecessary recursive calls. We give a formal description of the algorithm in Figure \ref{fig:alg}. We refer to the algorithm as $RandBinarySubIso$.

\begin{figure}[t]
\removelatexerror
\begin{center}
\parbox{\columnwidth}{
\SetAlgoFuncName{Algorithm}{anautorefname}
\begin{function}[H]
\caption{RandBinarySubIso($H,G$)}
\begin{enumerate}
\item If $|H|=0$, return \textbf{true};
\item If $|G|=0$, return \textbf{false};
\item With probability $1/2$ swap $H_L$ and $H_R$ in $H$;
\item With probability $1/2$ swap $G_L$ and $G_R$ in $G$;
\item If $RandBinarySubIso(H_L,G_L) = \textbf{false}$, then go to step 7;
\item If $RandBinarySubIso(H_R,G_R) = \textbf{true}$, then return \textbf{true};
\item If $RandBinarySubIso(H_L,G_R) = \textbf{false}$, then return \textbf{false};
\item If $RandBinarySubIso(H_R,G_L) = \textbf{true}$, then return \textbf{true}. \\Otherwise return \textbf{false};
\end{enumerate}
\vspace{-0.2cm}
\end{function}
}
\end{center}
\caption{A randomized, recursive algorithm for rooted Subtree Isomorphism on binary trees.}
\label{fig:alg}
\end{figure}

\begin{theorem}\label{heightthm_re}
The $RandBinarySubIso$ algorithm runs in expected time $O(\min\{2.8431^h,n^2\})$ for trees $H$ and $G$ of size $O(n)$ and height at most $h$. In particular, it runs in time $O(n^{1.507})$ for trees of height $(1+o(1))\cdot \log_2{n}$, and is strongly subquadratic for trees of height $h<1.3267\log_2{n}$.
\end{theorem}

Before proving Theorem \ref{heightthm_re} we first prove a useful lemma.
Let $T(h)$ be the maximum expected number of times $RandBinarySubIso(H,G)$ makes a recursive call with an empty tree when $H$ and $G$ are arbitrary rooted trees with height at most $h$. Let $T_{yes}(h)$ and $T_{no}(h)$ be defined similarly, but under the assumption that the algorithm returns \textbf{true} and \textbf{false}, respectively. Note that $T(0) = T_{yes}(0) = T_{no}(0) = 1$. Also note that $T(h) = \max\{T_{yes}(h),\; T_{no}(h)\}$.

\begin{lemma} \label{rec}
For all $h \geq 0$,
\begin{align*}
T_{yes}(h)&~\leq~ 2.25\cdot T_{yes}(h-1) + 0.5\cdot T_{no}(h-1)~,\\
T_{no}(h)&~\leq~ T_{yes}(h-1) + 2\cdot T_{no}(h-1)~.
\end{align*}
\end{lemma}
\begin{proof}
To simplify notation we write $H \subseteq G$ when $H$ is isomorphic to a subtree of $G$, and $H \not\subseteq G$ otherwise.

We first show that $T_{yes}(h)\leq 2.25\cdot T_{yes}(h-1) + 0.5\cdot T_{no}(h-1)$. Assume therefore that $H \subseteq G$. With probability $1/2$ we then have $H_L \subseteq G_L$ and $H_R \subseteq G_R$, such that the algorithm returns \textbf{true} in line 6 after spending $2\cdot T_{yes}(h-1)$ time in expectation. On the other hand, with probability $1/2$ the outcomes of lines 5 and 6 depend on the trees in question, and the recursive calls in lines 7 and 8 both return \textbf{true} if reached. More precisely, we get three cases that depend on the trees:
\begin{itemize}
\item[$(i)$]
$H_L \subseteq G_L$ and $H_R \subseteq G_R$: The recursive calls in lines 5 and 6 both return \textbf{true}, and the algorithm spends $2\cdot T_{yes}(h-1)$ time in expectation.
\item[$(ii)$]
$H_L \not\subseteq G_L$ and $H_R \not\subseteq G_R$: The recursive call in line 5 returns \textbf{false}, and the recursive calls in lines 7 and 8 both return \textbf{true}. The algorithm spends $T_{no}(h-1)+2\cdot T_{yes}(h-1)$ time in expectation.
\item[$(iii)$]
$H_L \subseteq G_L$ and $H_R \not\subseteq G_R$, or $H_L \not\subseteq G_L$ and $H_R \subseteq G_R$: The recursive call in line 5 returns \textbf{false} with probability $1/2$ and \textbf{true} with probability $1/2$. In the second case the recursive call in line 6 returns \textbf{false}. The recursive calls in lines 7 and 8 both return \textbf{true}. The algorithm spends $T_{no}(h-1)+2.5\cdot T_{yes}(h-1)$ time in expectation.
\end{itemize}
The third case thus dominates the two others, and we conclude that $T_{yes}(h)\leq 2.25\cdot T_{yes}(h-1) + 0.5\cdot T_{no}(h-1)$.

We next show that $T_{no}(h)\leq T_{yes}(h-1) + 2\cdot T_{no}(h-1)$. Assume therefore that $H \not\subseteq G$. We get the contribution $2\cdot T_{no}(h-1)$ as follows. In either line $5$ or $6$ we get the answer \textbf{false} from a recursive call, and in either line $7$ or $8$ we also get the answer \textbf{false} from a recursive call. This amounts to two ``no'' answers which cost $2\cdot T_{no}(h-1)$ in expectation. We get the contribution $T_{yes}(h-1)$ as follows. With probability at most $1/2$ we get the answer \textbf{true} in line $5$ (which means that we get \textbf{false} in line 6). Similarly, with probability at most $1/2$ we get the answer \textbf{true} in line $7$ (which means that we get \textbf{false} in line 8). In total, we get that $T_{no}(h)\leq 2\cdot T_{no}(h-1)+{1 \over 2}T_{yes}(h-1)+{1 \over 2}T_{yes}(h-1)$.
\end{proof}

\begin{proof}[Proof of Theorem~\ref{heightthm_re}]
Lemma \ref{rec} gives us that
\begin{align*}
\left( 	\begin{array}{c}
T_{yes}(h) \\
T_{no}(h)
\end{array}
\right)
&~\leq~
\left( 	\begin{array}{cc}
2.25 & 0.5 \\
1 & 2
\end{array}
\right)
\left( 	\begin{array}{c}
T_{yes}(h-1) \\
T_{no}(h-1)
\end{array}
\right)\\
&~\leq~
\left( 	\begin{array}{cc}
2.25 & 0.5 \\
1 & 2
\end{array}
\right)^h
\left( 	\begin{array}{c}
1 \\
1
\end{array}
\right) ~.
\end{align*}
A diagonalization of the matrix yields
	$$
		\left( 	\begin{array}{cc}
			2.25 & 0.5 \\
			1 & 2
			\end{array}
		\right)
		~=~ Q^{-1}JQ~,
$$
where
\begin{align*}
Q^{-1} &~=~ \left( 	\begin{array}{cc}
{1 - \sqrt{33} \over 8} & {1 - \sqrt{33} \over 8} \\
1 & 1
\end{array}
\right)\\
J &~=~
\left( 	\begin{array}{cc}
{17 - \sqrt{33} \over 8} & 0 \\
0 & {17 + \sqrt{33} \over 8}
\end{array}
\right)\\
Q &~=~
\left( 	\begin{array}{cc}
-{4\over \sqrt{33}} & {1 \over 2}+{1 \over 2\sqrt{33}} \\
{4\over \sqrt{33}} & {1 \over 2}-{1 \over 2\sqrt{33}}
\end{array}
\right) ~,
\end{align*}
	and therefore
\[
\left(\begin{array}{c}
        \!T_{yes}(h)\! \\
	\!T_{no}(h)\!
\end{array}\right)
\,\leq\,
\left(\begin{array}{c}
        \!0.065\cdot 1.407^h+0.94 \cdot 2.8431^h\! \\
	\!-0.109 \cdot 1.407^h+1.109 \cdot 2.8431^h\!
\end{array}\right)\,.
\]
Thus,
$T(h)=O(2.8431^h)$, which proves the theorem.
\end{proof}

\subsection{A Faster Algorithm for Ternary Trees}
Here we discuss the subtree isomorphism problem for rooted ternary trees. We prove Theorem~\ref{thm:deg3} by showing that Subtree isomorphism for rooted ternary trees of height $h$ can be solved in expected time $O(6.107^h)$.
Just as with the binary case, this running time is lower than the runtime given by our generic algorithm for constant degree trees in Section \ref{sec:any_upper}.

Similarly to the binary case, the proof of the theorem proceeds by a recursive approach. In each recursive call, we consider a randomized decision tree for $3\times 3$ bipartite perfect matching, where each query corresponds to a recursive call on height one less. We then analyze the runtime similar to the binary tree case: we distinguish between the ``yes'' and ``no'' case of the query answer, and write the running time as two recurrences, one for $T_{yes}$, when the algorithm said the trees are isomorphic, and one for $T_{no}$ when they were not.
We analyze the randomized decision tree in terms of the expected number of ``yes'' and ``no'' query answers in the worst case.

The randomized query protocol is as follows. Let $U$ and $V$ be the two partitions of the bipartite matching instance (respectively, $U$ are the subtrees of the root of one tree and $V$ are the subtrees of the root of the other). First we pick $U$ or $V$ at random w.p. $1/2$. If we pick $V$, then the names of $U$ and $V$ are swapped. Now, with probability $1/6$ we pick a permutation of the vertices in $U$, and with probability $1/6$ we pick a permutation of $V$. After these two permutations are fixed, the protocol is deterministic.
Let $a,b,c$ be the nodes of $U$ and $x,y,z$ be the nodes of $V$, in the order of the chosen permutations.
The deterministic decision tree we use is depicted in Figures~\ref{fig:bigtree} and~\ref{fig:smalltrees}.

\begin{figure}
\begin{center}
\includegraphics[scale=0.75]{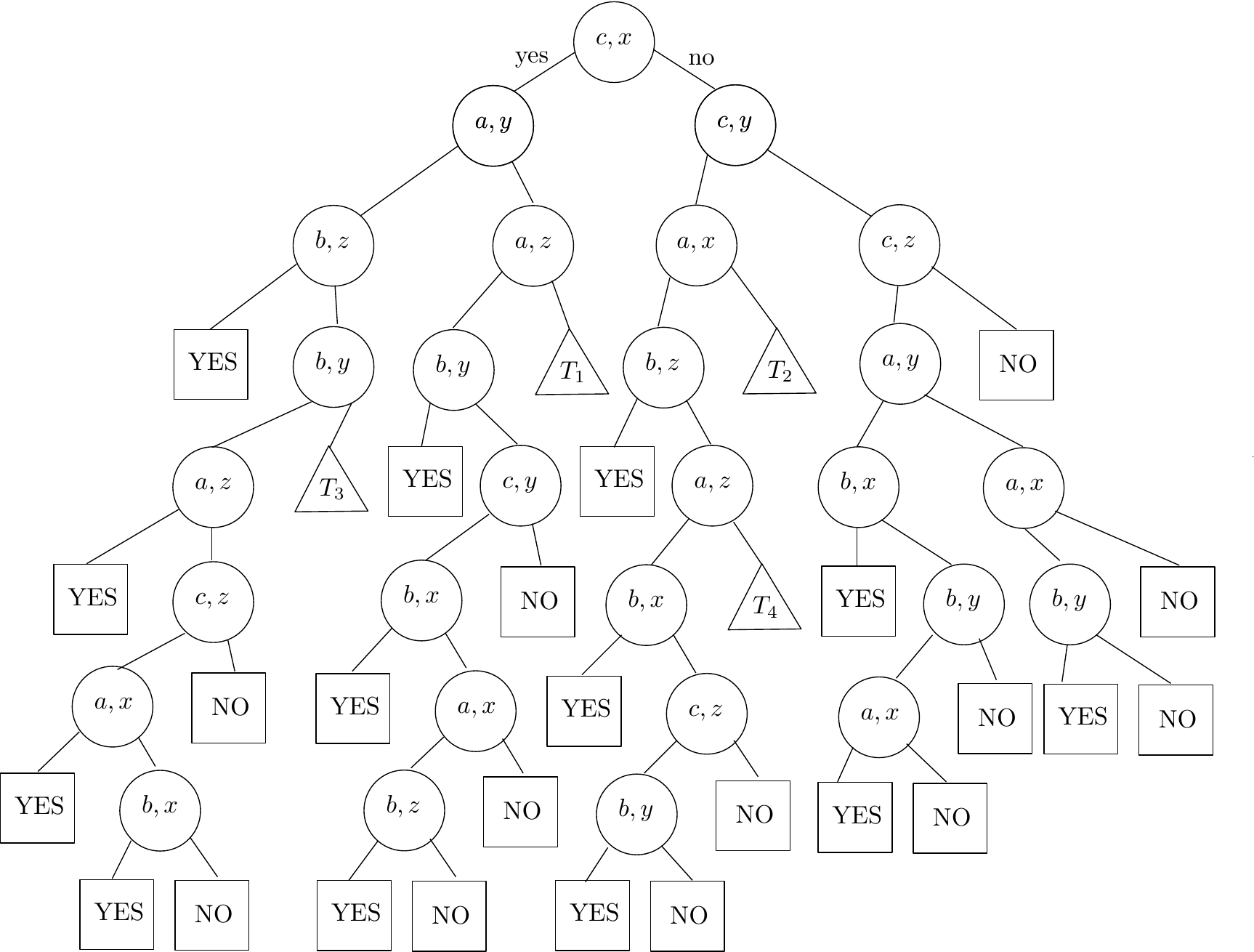}
\end{center}
\caption{The decision tree used for bipartite matching in the degree $3$ case.}
\label{fig:bigtree}
\end{figure}

\begin{figure}
\vspace{0.5cm}
\begin{center}
\includegraphics[scale=0.75]{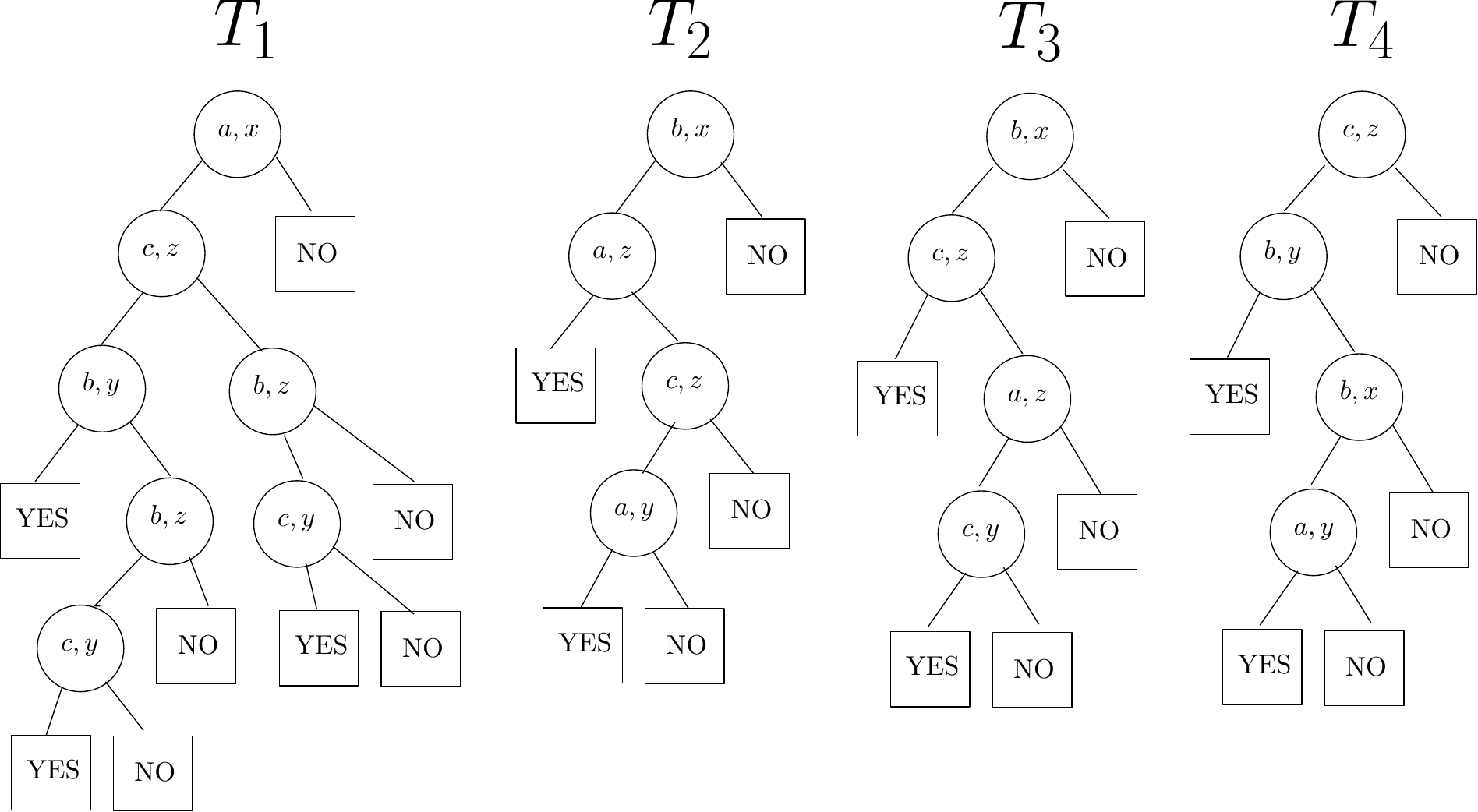}
\end{center}
\caption{The missing subtrees of the decision tree used for bipartite matching in the degree $3$ case.}
\label{fig:smalltrees}
\end{figure}

For each of the $2^9$ choices for the answers to the $9$ edge queries in the $3\times 3$ matching instance, we consider each of the $72$ randomized choices as described above (swap $U$ and $V$, permute $U$ and $V$) and consider the decision tree, computing the expected number of ``yes'' and ``no'' calls. Using a computer program, we establish that when the instance has no perfect matching, the expected number of ``yes'' calls is always at most $26/9$, and the expected number of ``no'' calls is always at most $37/9$; this happens when the complement of the graph consists of a 4-cycle, disjoint from a single edge. On the other hand, if the instance has a perfect matching there are two cases that dominate all others: when the expected number of ``yes'' calls is $131/36$, and the expected number of ``no'' calls is $61/36$, or when the expected number of ``yes'' calls is $133/36$, and the expected number of ``no'' calls is $5/3$. There are thus two options for the recurrence relation, and one of them dominates the other. We present the recurrence that achieves the maximum, and hence gives the worst-case expected runtime for the ternary case.
$$
\left( 	\begin{array}{c}
T_{yes}(h) \\
T_{no}(h)
\end{array} 
\right)
~\leq~
\left( 	\begin{array}{cc}
133/36 & 5/3 \\
26/9 & 37/9
\end{array} 
\right)
\left( 	\begin{array}{c}
T_{yes}(h-1) \\
T_{no}(h-1)
\end{array} 
\right)
~\leq~
		\left( 	\begin{array}{cc}
			133/36 & 5/3 \\
			26/9 & 37/9
			\end{array} 
		\right)^h
		\left( 	\begin{array}{c}
			1 \\
			1
			\end{array} 
		\right)
$$

	The diagonalization yields
	$$
		\left( 	\begin{array}{cc}
			133/36 & 5/3 \\
			26/9 & 37/9
			\end{array} 
		\right)
		~=~ Q^{-1} J Q,$$
		
		where 
\begin{align*}
	J &~=~	\left( 	\begin{array}{cc}
			\frac{281-\sqrt{25185}}{72} & 0 \\
			0 & {\frac{281+\sqrt{25185}}{72}}
			\end{array} 
		\right)\\
	Q &~=~	\left( 	\begin{array}{cc}
			\frac{1}{2}+ \frac{\sqrt{25185}}{3358} & \frac{1}{2}- \frac{\sqrt{25185}}{3358} \\
			-\frac{104}{\sqrt{25185}} & \frac{104}{\sqrt{25185}}
			\end{array} 
		\right)~,		
\end{align*}
	which gives that 
\[
\left(\begin{array}{c}
        T_{yes}(h) \\
	T_{no}(h)
\end{array}\right)
~\leq~
\left(\begin{array}{c}
        0.17\cdot 1.7^h+ 0.831 \cdot 6.107^h \\
	-0.2 \cdot 1.69^h+1.21\cdot 6.107^h
\end{array}\right)~.
\]

Thus, the running time overall is $O(6.107^h)$.


\subsection{An algorithm for any constant degree}\label{sec:any_upper}

In this section we describe a way to use randomization to save subtree comparisons in the Edmonds-Matula algorithm \cite{Matula1978} for all degrees $d > 2$. Recall that the algorithm works as follows.
Given two trees $H$ and $G$ of constant degree $d$, the goal is to decides whether $H$ is isomorphic to a subtree of $G$ by using recursion. If the roots of either $H$ or $G$ have less than $d$ children, we simply view the missing subtrees as being a special empty subtree.
\begin{enumerate}
\item Let $H_1,\ldots,H_d$ be the $d$ subtrees of $H$, and let $G_1,\ldots,G_d$ be the $d$ subtrees of $G$;
\item Build a bipartite graph $\mathcal{G}$ with $d$ vertices $\mathcal{U} = \{u_1,\ldots,u_d\}$ on the left and $d$ vertices $\mathcal{W} = \{v_1,\ldots,v_d\}$ on the right. For all $i,j \in [d]$, connect $u_i$ and $v_j$ if and only if $H_i$ is isomorphic to a subtree of $G_j$. We decide which edges appear in the graph recursively.
\item Output that $H$ is isomorphic to a subtree of $G$ if and only if there is a perfect matching in the bipartite graph $\mathcal{G}$.
\end{enumerate}
The runtime of the algorithm is $O(\min\{d^{2h},n^2\})$, where $h$ is the height.
Intuitively, we can improve the runtime of the algorithm as follows. Perform recursive calls corresponding to edges $(u_i,v_j)$ in a random order, and stop as soon as we either detect a perfect matching or rule out the existence of a perfect matching. It is not difficult to show that this randomized version of the algorithm performs $d^2-\Omega(1)$ recursive calls in expectation out of the $d^2$ possible calls. That is, in expectation, we save at least a constant number of recursive calls. This implies that the algorithm runs in $O((d^2-\Omega(1))^h)$ expected time, which is faster than the deterministic algorithm. However, we prove below that we can save $\Omega(d)$ recursive calls in expectation using a slightly different variant of the randomized algorithm.

\begin{lemma} \label{query}
Let $\mathcal{G}$ be a bipartite graph with $d$ vertices $\mathcal{U} = \{u_1,\ldots,u_d\}$ on the left and $d$ vertices $\mathcal{W} = \{v_1,\ldots,v_d\}$ on the right, and suppose we are given query access to the adjacency matrix of $\mathcal{G}$. There is a randomized query algorithm that decides whether $\mathcal{G}$ contains a perfect matching by making $d^2-{1 \over 3}d+{2 \over 3}$ queries in expectation, with probability 0 of making an error.
\end{lemma}

We use the following two claims to prove the lemma.

	\begin{claim} \label{yescase}
		Assume that $\mathcal{G}$ has a perfect matching. Then the following algorithm finds a perfect matching after making $d^2-d+2$ expected queries: Query edges $(u_i,v_j)$ in a random order, and stop when finding a perfect matching.
	\end{claim}
	\begin{proof}
Fix a perfect matching present in $\mathcal{G}$ and call its $d$ edges ``marked''.
We stop when all marked edges have been queried. There are $d^2-d$ unmarked edges. The probability that a given unmarked edge is not queried is $1 \over d+1$. Therefore, the expected number of unqueried, unmarked edges is ${d^2-d \over d+1}\geq d-2$.
	\end{proof}
	\begin{claim} \label{nocase}
		Assume that $\mathcal{G}$ does not have a perfect matching. Then the following algorithm makes at most $d^2-{1 \over 2}d+1$ queries in expectation before determining that $\mathcal{G}$ does not contain a perfect matching.
		\begin{enumerate}
			\item With probability $1/2$ swap $\mathcal{U}$ and $\mathcal{W}$;
			\item Randomly permute the vertices of $\mathcal{U}=\{u_1,u_2,\ldots,u_d\}$;
			\item Query all edges adjacent to $u_i$ for $i$ going from 1 to $d$, but stop when ruling out the existence of a perfect matching, i.e., stop when the set of processed vertices $S=\{u_1,\ldots,u_i\}$ contains a subset $S'$ with a neighbourhood $N(S')$ that is smaller than the size of $S'$.
		\end{enumerate}
	\end{claim}
	\begin{proof}
		Consider the sets $\mathcal{U}$ and $\mathcal{W}$ prior to running the algorithm.
		By Hall's theorem, the set $\mathcal{U}$ contains a set $S'$ such that $|N(S')|<|S'|$. We can assume that $|S'|=|N(S')|+1$, since otherwise we can iteratively remove a vertex from $S'$ until this condition is satisfied. Consider two cases.
		\begin{itemize}
			\item $d$ is even: If $|S'|\geq {d \over 2}+1$, we define $T'=\mathcal{W} \setminus N(S')$. Because $N(S')\geq d/2$, we get that $|T'|\leq {d \over 2}$. By our construction of $T'$, we have that $N(T') \subseteq \mathcal{U} \setminus S'$ and, as a result, $|N(T')|<|T'|$. Given the first step of the algorithm, with probability at least $1/2$ the set $\mathcal{U}$ therefore contains a set $S'$ such that $|N(S')|<|S'|\leq {d \over 2}$.
			\item $d$ is odd: It follows from as similar argument that, with probability at least $1/2$, the set $\mathcal{U}$ contains a set $S'$ such that $|N(S')|<|S'|\leq {d+1 \over 2}$.
		\end{itemize}
		We now condition on the set $\mathcal{U}$ containing $S'$ with $|S'|\leq {d+1 \over 2}$ and $|N(S')|<|S'|$.

		The algorithm stops once it queries all vertices from $S'$, since a perfect matching is then ruled out by Hall's theorem. The probability that we do not process a given vertex before processing all vertices in $S'$ is ${1/(|S'|+1)}$. Therefore the expected number of unprocessed vertices when the algorithm stops is at least
		$$
			(d-|S'|)\cdot {1 \over |S'|+1}~\geq~ {d-1 \over 2}\cdot {1 \over {d+1 \over 2}+1}~=~{d-1 \over d+3}~.
		$$

		Hence, with probability $1/2$, we query $d \left(d-{d-1 \over d+3}\right)$ edges, and overall the number of queried edges is
$$
{1 \over 2}\left[d \left(d-{d-1 \over d+3}\right)\right]+{1 \over 2}d^2~=~d^2\left(1-{1-{1 \over d} \over 2(d+3)}\right) 
~\leq~ d^2-{1 \over 2}d+1~.
$$
		In the last inequality we use that $d \geq 3$.
	\end{proof}

\begin{proof}[Proof of Lemma \ref{query}]

	We prove the lemma by using claims \ref{yescase} and \ref{nocase}.

	With probability $1/3$ we run the algorithm from Claim \ref{yescase} and with probability $2/3$ we run the algorithm from Claim \ref{nocase}.
	Consider the case when $\mathcal{G}$ has a perfect matching. Then the expected number of  edges queried is upper bounded by
	$$
		{1 \over 3}(d^2-d+2)+{2 \over 3}d^2~=~d^2-{1 \over 3}d+{2 \over 3}~.
	$$
	On the other hand, for the case when $\mathcal{G}$ does not contain a perfect matching, the expected number of edges queried is upper bounded by
	$$
		{1 \over 3}d^2+{2 \over 3}\left(d^2-{1 \over 2}d+1\right)~=~d^2-{1 \over 3}d+{2 \over 3}~.
	$$
	Overall, regardless of $\mathcal{G}$, we therefore query at most $d^2-{1 \over 3}d+{2 \over 3}$ edges in expectation.
\end{proof}

\begin{theorem} \label{constanddegree_re}
	There is a randomized algorithm that solves Subtree Isomorphism on two rooted trees of size $O(n)$, constant degree $d$, and height at most $h$ in expected time $O\left(\left(d^2-\frac{1}{3} d+{2 \over 3}\right)^h \right)$.
	In particular, the algorithm is strongly subquadratic for trees of height
	$$
		h~\leq~ \left(\frac{2 \log d}{\log(d^2-\frac{1}{3} d+{2 \over 3})}-\epsilon\right) \cdot \log_d{n}~,
	$$
	for any constant $\epsilon>0$.
\end{theorem}
\begin{proof}
	We run the following randomized, recursive algorithm that decides whether $H$ is isomorphic to a subtree of $G$.
\begin{enumerate}
\item Let $H_1,\ldots,H_d$ be the $d$ subtrees of $H$, and let $G_1,\ldots,G_d$ be the $d$ subtrees of $G$;
\item Let $\mathcal{G}$ be a bipartite graph with $d$ vertices $\mathcal{U} = \{u_1,\ldots,u_d\}$ on the left and $d$ vertices $\mathcal{W} = \{v_1,\ldots,v_d\}$ on the right. For all $i,j \in [d]$, let $u_i$ and $v_j$ be connected if and only if $H_i$ is isomorphic to a subtree of $G_j$.
\item Decide whether the graph $\mathcal{G}$ has a perfect matching by running the algorithm from Lemma \ref{query}. Whenever we need to decide whether an edge $(u_i,v_j)$ is present in $\mathcal{G}$, do it recursively.
\end{enumerate}
By the proof of Lemma \ref{query}, it suffices to query $d^2-\frac{1}{3} d+{2 \over 3}$ edges for every level. Given that the height of the trees is upper bounded by $h$, we get the desired running time.
\end{proof}

\medskip
\paragraph{Acknowledgements.}
We would like to thank Shiri Chechik, Piotr Indyk, Haim Kaplan, Michael Kapralov, Huacheng Yu, and Uri Zwick for many helpful discussions.  

\bibliographystyle{abbrv}
\bibliography{ref}

\begin{thebibliography}{10}

\bibitem{ABV15}
A.~Abboud, A.~Backurs, and V.~{Vassilevska Williams}.
\newblock {Tight Hardness Results for {LCS} and other Sequence Similarity
  Measures}.
\newblock In {\em Proc.\ of the 56th FOCS}, 2015.

\bibitem{AGV15}
A.~Abboud, F.~Grandoni, and V.~V. Williams.
\newblock Subcubic equivalences between graph centrality problems, {APSP} and
  diameter.
\newblock In {\em Proc.\ of the 26th SODA}, pages 1681--1697, 2015.

\bibitem{AV14}
A.~Abboud and V.~{Vassilevska Williams}.
\newblock Popular conjectures imply strong lower bounds for dynamic problems.
\newblock {\em Proc.\ of the 55th FOCS}, pages 434--443, 2014.

\bibitem{AVW15}
A.~Abboud, V.~V. Williams, and J.~R. Wang.
\newblock Approximation and fixed parameter subquadratic algorithms for radius
  and diameter.
\newblock In {\em Proc.\ of the 27th SODA}, 2016.
\newblock To appear.

\bibitem{AVW14}
A.~Abboud, V.~V. Williams, and O.~Weimann.
\newblock Consequences of faster alignment of sequences.
\newblock In {\em Automata, Languages, and Programming}, pages 39--51.
  Springer, 2014.

\bibitem{AHM00}
T.~Akutsu and M.~M. Halld{\'o}rsson.
\newblock On the approximation of largest common subtrees and largest common
  point sets.
\newblock {\em Theoretical Computer Science}, 233(1):33--50, 2000.

\bibitem{ATMA14}
T.~Akutsu, T.~Tamura, A.~A. Melkman, and A.~Takasu.
\newblock On the complexity of finding a largest common subtree of bounded
  degree.
\newblock {\em Theoretical Computer Science}, 590:2--16, 2014.

\bibitem{ACRSZ10}
A.~Ambainis, A.~M. Childs, B.~Reichardt, R.~Spalek, and S.~Zhang.
\newblock Any {AND-OR} formula of size {N} can be evaluated in time
  {N}\({}^{\mbox{1/2+o(1)}}\) on a quantum computer.
\newblock {\em {SIAM} J. Comput.}, 39(6):2513--2530, 2010.

\bibitem{BI15}
A.~Backurs and P.~Indyk.
\newblock {Edit Distance Cannot Be Computed in Strongly Subquadratic Time
  (unless SETH is false)}.
\newblock In {\em Proc.\ of the 47th STOC}, pages 51--58, 2015.

\bibitem{Bille05}
P.~Bille.
\newblock A survey on tree edit distance and related problems.
\newblock {\em Theoretical Computer Science}, 337(1--3):217--239, 2005.

\bibitem{Bring}
K.~Bringmann.
\newblock Why walking the dog takes time: Fr{\'e}chet distance has no strongly
  subquadratic algorithms unless seth fails.
\newblock In {\em Proc.\ of the 55th FOCS}, pages 661--670, 2014.

\bibitem{BK15}
K.~Bringmann and M.~Kunnemann.
\newblock {Quadratic Conditional Lower Bounds for String Problems and Dynamic
  Time Warping}.
\newblock In {\em Proc.\ of the 56th FOCS}, 2015.

\bibitem{Che97}
J.~Cheriyan.
\newblock Randomized {O}({M}(\textbar{V}\textbar)) algorithms for problems in
  matching theory.
\newblock {\em SIAM Journal on Computing}, 26(6):1635--1655, 1997.

\bibitem{Chung1978}
M.~J. Chung.
\newblock {O}(n\textsuperscript{2.5}) time algorithms for the subgraph
  homeomorphism problem on trees.
\newblock {\em Journal of Algorithms}, 8(1):106--112, 1987.

\bibitem{CH97}
R.~Cole and R.~Hariharan.
\newblock Tree pattern matching and subset matching in randomized {O}(n
  log\({}^{\mbox{3}}\)m) time.
\newblock In {\em Proc.\ of the 29th STOC}, pages 66--75, 1997.

\bibitem{CH02}
R.~Cole and R.~Hariharan.
\newblock Verifying candidate matches in sparse and wildcard matching.
\newblock In {\em Proc.\ of the 34th STOC}, pages 592--601, 2002.

\bibitem{CH03}
R.~Cole and R.~Hariharan.
\newblock Tree pattern matching to subset matching in linear time.
\newblock {\em SIAM Journal on Computing}, 32(4):1056--1066, 2003.

\bibitem{CPS15}
M.~Cygan, J.~Pachocki, and A.~Socala.
\newblock The hardness of subgraph isomorphism.
\newblock {\em CoRR}, abs/1504.02876, 2015.

\bibitem{DLP00}
A.~Dessmark, A.~Lingas, and A.~Proskurowski.
\newblock Faster algorithms for subgraph isomorphism of k-connected partial
  k-trees.
\newblock {\em Algorithmica}, 27(3-4):337--347, 2000.

\bibitem{DIR99}
Y.~Dinitz, A.~Itai, and M.~Rodeh.
\newblock On an algorithm of zemlyachenko for subtree isomorphism.
\newblock {\em Information Processing Letters}, 70(3):141--146, 1999.

\bibitem{DGM94}
M.~Dubiner, Z.~Galil, and E.~Magen.
\newblock Faster tree pattern matching.
\newblock {\em Journal of the ACM (JACM)}, 41(2):205--213, 1994.

\bibitem{legallmult}
F.~L. Gall.
\newblock Powers of tensors and fast matrix multiplication.
\newblock In {\em Proc.\ of the 39th ISSAC}, pages 296--303, 2014.

\bibitem{Gallagher06}
B.~Gallagher.
\newblock Matching structure and semantics: A survey on graph-based pattern
  matching.
\newblock {\em AAAI FS}, 6:45--53, 2006.

\bibitem{GJ}
M.~R. Garey and D.~S. Johnson.
\newblock {\em Computers and intractability}, volume~29.
\newblock W. H. Freeman, 2002.

\bibitem{GMN15}
A.~C. Giannopoulou, G.~B. Mertzios, and R.~Niedermeier.
\newblock Polynomial fixed-parameter algorithms: {A} case study for longest
  path on interval graphs.
\newblock In {\em Proc. of IPEC}, 2015.
\newblock To appear.

\bibitem{GMKS90}
P.~B. Gibbons, R.~M. Karp, G.~L. Miller, and D.~Soroker.
\newblock Subtree isomorphism is in random {NC}.
\newblock {\em Discrete Applied Mathematics}, 29(1):35--62, 1990.

\bibitem{gusfield}
D.~Gusfield.
\newblock {\em Algorithms on strings, trees and sequences: Computer Science and
  Computational Biology}.
\newblock Cambridge University Press, 1997.

\bibitem{HO82}
C.~M. Hoffmann and M.~J. O'Donnell.
\newblock Pattern matching in trees.
\newblock {\em Journal of the ACM (JACM)}, 29(1):68--95, 1982.

\bibitem{HT72}
J.~E. Hopcroft and R.~E. Tarjan.
\newblock Isomorphism of planar graphs.
\newblock In {\em Proc.\ of Complexity of Computer Computations}, pages
  131--152. 1972.

\bibitem{IP01}
R.~Impagliazzo and R.~Paturi.
\newblock On the complexity of k-{SAT}.
\newblock {\em Journal of Computer and System Sciences}, 62(2):367--375, 2001.

\bibitem{IPZ01}
R.~Impagliazzo, R.~Paturi, and F.~Zane.
\newblock Which problems have strongly exponential complexity?
\newblock {\em Journal of Computer and System Sciences}, 63:512--530, 2001.

\bibitem{Indyk97}
P.~Indyk.
\newblock Deterministic superimposed coding with applications to pattern
  matching.
\newblock In {\em Proc.\ of the 38th FOCS}, pages 127--136, 1997.

\bibitem{Indyk98}
P.~Indyk.
\newblock Faster algorithms for string matching problems: Matching the
  convolution bound.
\newblock In {\em Proc.\ of the 39th FOCS}, pages 166--173, 1998.

\bibitem{Karp72}
R.~M. Karp.
\newblock Reducibility among combinatorial problems.
\newblock In {\em Complexity of Computer Computations}, The IBM Research
  Symposia Series, pages 85--103. Springer US, 1972.

\bibitem{KMY95}
S.~Khanna, R.~Motwani, and F.~F. Yao.
\newblock Approximation algorithms for the largest common subtree problem.
\newblock Technical report, Stanford University, 1995.

\bibitem{KM95}
P.~Kilpel{\"a}inen and H.~Mannila.
\newblock Ordered and unordered tree inclusion.
\newblock {\em SIAM Journal on Computing}, 24(2):340--356, 1995.

\bibitem{Kosaraju89}
S.~R. Kosaraju.
\newblock Efficient tree pattern matching (preliminary version).
\newblock In {\em Proc.\ of the 30th FOCS}, pages 178--183, 1989.

\bibitem{Lingas83}
A.~Lingas.
\newblock An application of maximum bipartite c-matching to subtree
  isomorphism.
\newblock In {\em Proc.\ of the 8th CAAP}, pages 284--299, 1983.

\bibitem{Lingas89}
A.~Lingas.
\newblock Subgraph isomorphism for biconnected outerplanar graphs in cubic
  time.
\newblock {\em Theoretical Computer Science}, 63(3):295--302, 1989.

\bibitem{LK89}
A.~Lingas and M.~Karpinski.
\newblock Subtree isomorphism is {NC} reducible to bipartite perfect matching.
\newblock {\em Information Processing Letters}, 30(1):27--32, 1989.

\bibitem{LP}
L.~Lov{\'a}sz and M.~D. Plummer.
\newblock {\em Matching theory}, volume 367.
\newblock American Mathematical Soc., 2009.

\bibitem{MP14}
D.~Marx and M.~Pilipczuk.
\newblock Everything you always wanted to know about the parameterized
  complexity of subgraph isomorphism (but were afraid to ask).
\newblock In {\em Proc.\ of the 31st STACS}, pages 542--553, 2014.

\bibitem{MT92}
J.~Matou{\v{s}}ek and R.~Thomas.
\newblock On the complexity of finding iso-and other morphisms for partial
  k-trees.
\newblock {\em Discrete Mathematics}, 108(1):343--364, 1992.

\bibitem{Matula1968}
D.~W. Matula.
\newblock An algorithm for subtree identification.
\newblock {\em SIAM Review}, 10:273--274, 1968.

\bibitem{Matula1978}
D.~W. Matula.
\newblock Subtree isomorphism in {O}(n\textsuperscript{5/2}).
\newblock In {\em Algorithmic Aspects of Combinatorics}, volume~2 of {\em
  Annals of Discrete Mathematics}, pages 91--106. Elsevier, 1978.

\bibitem{MS04}
M.~Mucha and P.~Sankowski.
\newblock Maximum matchings via gaussian elimination.
\newblock In {\em Proc.\ of the 45th FOCS}, pages 248--255, 2004.

\bibitem{PW10}
M.~Patrascu and R.~Williams.
\newblock On the possibility of faster {SAT} algorithms.
\newblock In {\em Proc.\ of the 21st SODA}, volume~10, pages 1065--1075, 2010.

\bibitem{Reyner77}
S.~W. Reyner.
\newblock An analysis of a good algorithm for the subtree problem.
\newblock {\em SIAM Journal on Computing}, 6(4):730--732, 1977.

\bibitem{RV13}
L.~Roditty and V.~Vassilevska~Williams.
\newblock Fast approximation algorithms for the diameter and radius of sparse
  graphs.
\newblock In {\em Proc.\ of the 45th STOC}, pages 515--524, 2013.

\bibitem{SW86}
M.~Saks and A.~Wigderson.
\newblock Probabilistic boolean decision trees and the complexity of evaluating
  game trees.
\newblock In {\em Proc.\ of the 27th FOCS}, pages 29--38, 1986.

\bibitem{ST99}
R.~Shamir and D.~Tsur.
\newblock Faster subtree isomorphism.
\newblock {\em Journal of Algorithms}, 33(2):267--280, 1999.

\bibitem{Tai79}
K.-C. Tai.
\newblock The tree-to-tree correction problem.
\newblock {\em Journal of the ACM (JACM)}, 26(3):422--433, 1979.

\bibitem{Valiente13book}
G.~Valiente.
\newblock {\em Algorithms on trees and graphs}.
\newblock Springer Science \& Business Media, 2013.

\bibitem{Verma92}
R.~M. Verma.
\newblock Strings, trees, and patterns.
\newblock {\em Information Processing Letters}, 41(3):157--161, 1992.

\bibitem{W04}
R.~Williams.
\newblock A new algorithm for optimal 2-constraint satisfaction and its
  implications.
\newblock {\em Theoretical Computer Science}, 348(2):357--365, 2005.

\bibitem{v12}
V.~V. Williams.
\newblock Multiplying matrices faster than coppersmith-winograd.
\newblock In {\em Proc.\ of the 44th STOC}, pages 887--898, 2012.

\bibitem{ZJ94}
K.~Zhang and T.~Jiang.
\newblock Some max snp-hard results concerning unordered labeled trees.
\newblock {\em Information Processing Letters}, 49(5):249--254, 1994.

\bibitem{ZS89}
K.~Zhang and D.~Shasha.
\newblock Simple fast algorithms for the editing distance between trees and
  related problems.
\newblock {\em SIAM journal on computing}, 18(6):1245--1262, 1989.

\bibitem{ZSS92}
K.~Zhang, R.~Statman, and D.~Shasha.
\newblock On the editing distance between unordered labeled trees.
\newblock {\em Information processing letters}, 42(3):133--139, 1992.

\end{thebibliography}

\appendix

\section{Analysis of the Edmonds-Matula algorithm and its variants}\label{UpperBoundAppendix}
\begin{lemma}
On binary trees, the Edmonds-Matula algorithm takes $O(mn)$ time, where $m=|H|,\;n=|G|$.
\end{lemma}
\begin{proof}
Denote by $m_L,m_R,n_L,n_R$ the sizes of $H_L,H_R,G_L,G_R$, the left and right subtrees of $H$ and $G$, notice that $m_L+m_R=m-1,\;n_L+n_R=n-1$.
The runtime of the algorithm is described by the recurrence
\begin{align*}
T(0,n)&~=~T(m,0)~=~1 ~,\\
T(m,n) &~=~ 1 + T(m_L,n_L) + T(m_R,n_R) +T(m_L,n_R) + T(m_R,n_L)~.
\end{align*}
Then, by induction, we prove $T(m,n)\leq mn$,
\begin{align*}
T(m,n) &= 1 + T(m_L,n_L) + T(m_R,n_R) + \\
&\hspace{3cm}T(m_L,n_R) + T(m_R,n_L)\\
&\leq 1 + m_L\cdot n_L + m_R\cdot n_R + m_L\cdot n_R + m_R\cdot n_L\\
&= 1+(m_L+m_R)\cdot (n_L+n_R) \\
&= 1+(m-1)(n-1) \\
&\leq mn~.
\end{align*}
\end{proof}

As mentioned in section~\ref{sec:UB}, this algorithm is easily extended to solve the labelled version of the problem or the \emph{Largest Common Subtree} problem for any constant bounded degree $d=O(1)$. For completeness, we include pseudo-code of a variant that solves the \emph{Labelled Largest Common Subtree} problem, generalizing both.

\begin{algorithm}
\caption{$LLCS(H,G,d)$}
\label{MCS}
\begin{algorithmic}
\IF{$Size(F)=0$ \OR $Size(G)=0$}
\RETURN{0}
\ENDIF
\FOR{$i=1$ to $d$}
\FOR{$j=1$ to $d$}
\IF {$Label(H.Children[i])=Label(G.Children[j])$}
\STATE $Sub[i,j] \leftarrow LLCS(Subtree(H.Children[i]),Subtree(G.Children[j]),\;d)$
\ELSE
\STATE $Sub[i,j] \leftarrow 0$
\ENDIF
\ENDFOR
\ENDFOR
\STATE $w \leftarrow $ the weight of a maximum weight bipartite 
matching in the bipartite graph with 
\STATE edges defined by $Sub[i,j]$.
\RETURN {$w+1$}
\end{algorithmic}
\end{algorithm}

\begin{lemma}
Algorithm~\ref{MCS} solves the \emph{Labelled Largest Common Subtree} problem in time $O(mn)$ for rooted trees $H$,$G$ of bounded degree $d=O(1)$, where $m,n$ are the sizes of $H,G$ respectively.
\end{lemma}
\begin{proof}
Correctness is straightforward, it is also clear that as $d=O(1)$, Algorithm~\ref{MCS} makes a constant number of operations excluding the recursive calls. Denote by $m_1,m_2,...,m_r$ the sizes of the (maximal) subtrees rooted at the $r\leq d$ children of the root of $H$, and by $n_1,n_2,...,n_s$ the sizes of those rooted at the $s\leq d$ children of the root of $G$. It holds that $\sum_{i=1}^{r} m_i = m-1$ and $\sum_{j=1}^{s} n_j = n-1$.
The runtime of the algorithm is described by the recurrence
\begin{align*}
T(0,n)&~=~T(m,0)~=~1~,\\
T(m,n) &~=~ 1 + \sum_{i=1, j=1}^{r,s} T(m_i,n_j)~.
\end{align*}
Then, by induction, we prove $T(m,n)\leq mn$,
\begin{align*}
T(m,n) &~=~ 1 + \sum_{i=1, j=1}^{r,s} T(m_i,n_j)\\
&~\leq~ 1 + \sum_{i=1, j=1}^{r,s} m_i\cdot n_j\\
&~=~ 1+(\sum_{i=1}^{r} m_i)\cdot (\sum_{j=1}^{s} n_j) \\
&~=~ 1+(m-1)(n-1) \\
&~\leq~ mn~.
\end{align*}

\end{proof}

\end{document}